\documentclass[runningheads]{llncs}

\usepackage{graphicx}
\usepackage{amsfonts}
\usepackage{subcaption}
\usepackage{mathtools}

\newtheorem{observation}[theorem]{Observation}

\def\reals{{\mathbb R}}
\def\eps{{\varepsilon}}

\newcommand{\old}[1]{{{}}}

\begin{document}

\title{Spanners under the Hausdorff and Fr{\'e}chet Distances}
\titlerunning{Spanners under the Hausdorff and Fr{\'e}chet Distances}
\author{Tsuri Farhana \and
Matthew J. Katz}

\authorrunning{T. Farhana and M. Katz}
\institute{Ben Gurion University of the Negev, Beer Sheva, Israel 
\email{\{tsurif@post.bgu.ac.il,matya@bgu.ac.il}\}}

\maketitle              % typeset the header of the contribution

\begin{abstract}
We initiate the study of spanners under the Hausdorff and Fr{\'e}chet distances. 
We show that any $t$-spanner of a planar point-set $S$ is
a $\frac{\sqrt{t^2-1}}{2}$-Hausdorff-spanner and a $\min\{\frac{t}{2},\frac{\sqrt{t^2-t}}{\sqrt{2}}\}$-Fr{\'e}chet spanner.
We also prove that for any $t > 1$, there exist a set of points $S$ and an $\eps_1$-Hausdorff-spanner of $S$ and an $\eps_2$-Fr{\'e}chet-spanner of $S$, where $\eps_1$ and $\eps_2$ are constants, such that neither of them is a $t$-spanner.
\end{abstract}

\section{Introduction}
Let $S$ be a set of points in $\reals^d$. The \emph{Euclidean graph} over $S$, denoted $G_S$, is the complete graph over $S$, in which the weight of an edge $(u,v)$ is the Euclidean distance between its endpoints, denoted $d(u,v)$. A subgraph $H$ of $G_S$ is a {\em $t$-spanner} of $S$, for a real number $t \ge 1$, if it is a $t$-spanner of $G_S$, that is, if for any pair of points $u,v \in S$, the length of the shortest path between $u$ and $v$ in $H$ is at most $t\cdot d(u,v)$. In general, a path in $G_S$ between $u$ and $v$ whose length is at most $t\cdot d(u,v)$, is called a \emph{$t$-path}.

Geometric spanners, i.e., $t$-spanners of $G_S$, have been studied extensively over the years (see~\cite{NS09}), where the goal is often to construct a $t$-spanner, for a given $t > 1$, with some desirable properties, such as, small number of edges, small weight, small degree, and small diameter (i.e., the maximum number of edges in a minimum-hop $t$-path).

In this paper we initiate the study of Hausdorff-spanners and Fr{\'e}chet-spanners. 
Before defining these spanners, let us recall the definitions of the corresponding distances, i.e., Hausdorff~\cite{Hausdorff1914} (denoted $d_H$) and Fr\'echet~\cite{Frechet1906} (denoted $d_F$), adapted to our specific setting. The \emph{Hausdorff} distance between a segment $\overline{uv}$ and a polygonal path $P(u,v)$ between $u$ and $v$ is the maximum distance between a point $p$ on $P(u,v)$ and its closest point on $\overline{uv}$, that is, $d_H(P(u,v),\overline{uv}) = \max\{d(p,\overline{uv})\, |\, p \in P(u,v)\}$. The \emph{Fr\'echet} distance between $P(u,v)$ and $\overline{uv}$ is the minimum leash length, such that, a dog and its owner, both initially located at $u$, can walk along $P(u,v)$ and $\overline{uv}$, respectively, from $u$ to $v$, without backtracking.

We now define the notions of Hausdorff-spanner and Fr{\'e}chet-spanner. Let $\eps$ be a non-negative real number.  
A path $P(u,v)$ in $G_S$ between $u$ and $v$ is an \emph{$\eps$-Hausdorff-path} (\emph{$\eps$-Fr{\'e}chet-path}), if $d_H(P(u,v),\overline{uv}) \le \eps \cdot d(u,v)$
($d_F(P(u,v),\overline{uv}) \le \eps \cdot d(u,v)$).
Thus, a sub-graph $H$ of $G_S$ is an \emph{$\eps$-Hausdorff-spanner} (\emph{$\eps$-Fr{\'e}chet-spanner}), if 
there exists in $H$ an $\eps$-Hausdorff-path ($\eps$-Fr{\'e}chet-path) between any two points $u$ and $v$.

Let $S$ be a set of points in the plane. We show that a $t$-spanner is an $\eps_1$-Hausdorff-spanner, for $\eps_1=f(t)$, and an $\eps_2$-Fr{\'e}chet-spanner, for $\eps_2=f(t)$. We also prove that for any $t > 1$, there exist an $\eps_1$-Hausdorff-spanner $H_1$ and an $\eps_2$-Fr{\'e}chet-spanner $H_2$, where $\eps_1$ and $\eps_2$ are constants, such that both $H_1$ and $H_2$ are \emph{not} a $t$-spanner.

Finally, we show that a WSPD-based $t$-spanner is an $\eps$-Fr{\'e}chet-spanner, for $\eps=\frac{t-1}{2t+2}$. This bound is significantly better than the general one (mentioned above) when $\eps$ is close to zero. Since, to construct an $\eps$-Fr{\'e}chet-spanner, using an arbitrary $t$-spanner and thus the general bound, we need to set $t$ to $\frac{1+\sqrt{1+8\eps^2}}{2}$, while to do so using a WSPD-based $t$-spanner, we need to set $t$ to $\frac{1+2\eps}{1-2\eps}$. Notice that when $\eps$ is close to zero, the latter expression is much larger than the former one, which may be an advantage, since some of the spanner's properties, such as its number of edges, depend on $t$.

%%%%%%%%%%%%%%%%%%%%%%%%%%%%
%%%%% Hausdorff-spanners
%%%%%%%%%%%%%%%%%%%%%%%%%%%%
\section{Hausdorff spanners}
\label{sec:Hausdorff}

\subsection{$t$-spanners are $\eps$-Hausdorff-spanners}

In this section we show that a $t$-spanner is an $\eps$-Hausdorff-spanner, for $\eps = f(t)$.

Let $H$ be a $t$-spanner, $t > 1$, and let $u,v \in S$. We assume, without loss of generality, that $d(u,v)=1$ and that $u=(0,0)$ and $v=(1,0)$. Let $P(u,v)$ be a $t$-path in $H$ between $u$ and $v$. Then, $P(u,v)$'s length is at most $t$.

\begin{figure}[h]
    \centering
    \includegraphics[width=0.5\linewidth]{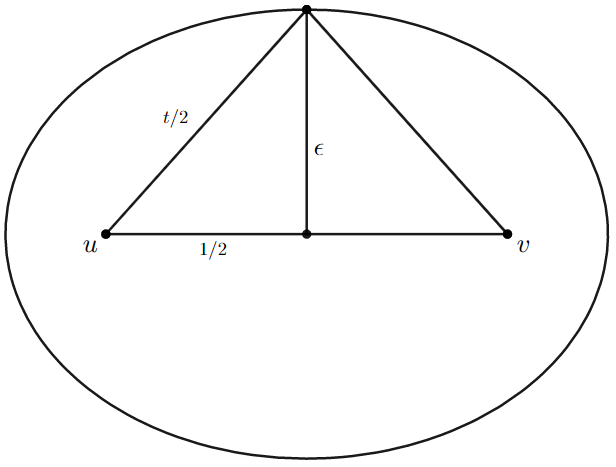}
    \caption{The furthest possible distance from a $t$-path to the segment between its endpoints.}
    \label{fig:max_haudorff_length}
\end{figure}

\begin{observation}
    \label{obs:t-path_in_ellipse}
    Let $E$ be the ellipse with foci points at $u$ and $v$, for which $d(p,u)+d(p,v)=t \cdot d(u,v)=t$, for every point $p$ on its boundary.  Then, $P(u,v) \subset E$. Moreover, the points of $E$ that are furthest from $\overline{uv}$ are the boundary points above and below the midpoint of $\overline{uv}$; their distance to $\overline{uv}$ is $\frac{\sqrt{t^2-1}}{2}$ (see Figure~\ref{fig:max_haudorff_length}).
\end{observation}

\begin{lemma}
\label{lem:bound_for_hausdorff}
$P(u,v)$ is a $\frac{\sqrt{t^2-1}}{2}$-Hausdorff-path.
\end{lemma}
\begin{proof}
Since $d(u,v)=1$, we need to show that $d_H(P(u,v),\overline{uv}) \le \frac{\sqrt{t^2-1}}{2}$.  
On the one hand, by Observation~\ref{obs:t-path_in_ellipse}, the distance from any point $p$ on $P(u,v)$ to $\overline{uv}$ is at most $\frac{\sqrt{t^2-1}}{2}$. On the other hand, let $q=(q_x,0)$ be a point on $\overline{uv}$. Then, since $P(u,v)$ is a path between $u$ and $v$, there exists a point $p$ on $P(u,v)$ with $x$-coordinate $q_x$. Now, by Observation~\ref{obs:t-path_in_ellipse}, the distance from $q$ to $p$
%(and thus to $P(u,v)$)
is at most $\frac{\sqrt{t^2-1}}{2}$. We conclude that $d_H(P(u,v),\overline{uv}) \le \frac{\sqrt{t^2-1}}{2}$.  
\end{proof}

We have shown that any $t$-path in $H$ is a $\frac{\sqrt{t^2-1}}{2}$-Hausdorff-path, and therefore $H$ is a $\frac{\sqrt{t^2-1}}{2}$-Hausdorff-spanner. We conclude that

\begin{corollary}
Any t-spanner is an $\eps$-Hausdorff-spanner, for $\eps = \frac{\sqrt{t^2-1}}{2}$.

\end{corollary}

\subsection{$\eps$-Hausdorff-spanners are not necessarily $t$-spanners}
\label{sec:Hausdorff_not_Euclid}
In this section we show that not every Hausdorff spanner is a Euclidean  spanner. More precisely, we present an infinite sequence of graphs, such that all of them are $c$-Hausdorff spanners, for some fixed constant $c > 0$, but for any $t \ge 1$, there exists a graph in the sequence that is not a $t$-spanner. Formally, we prove the following theorem.

\begin{theorem}
There exists a constant $c > 0$, such that for any $t\ge 1$, one can construct a graph that is a $c$-Hausdorff-spanner and is not a $t$-spanner.
\label{thm:Hausdorff_not_Euclid}
\end{theorem}

To prove the theorem, we define a sequence of graphs $F_0, F_1, F_2 \ldots$ as follows:
Let $F_0$ be the unit line segment with endpoints $(0,0)$ and $(1,0)$, that is, the endpoints are the vertices of $F_0$ and the segment is its single edge. Now, for any $n > 0$, we construct $F_n$ from $F_{n-1}$, by considering each segment (i.e., edge) $s$ of $F_{n-1}$ and (i) partitioning $s$ into three subsegments of equal length by adding two new vertices, (ii) forming an equilateral triangle with the middle subsegment as its base by adding a vertex on the outer side of the middle subsegment and connecting it to the endpoints of the middle subsegmet, and (iii) removing the middle subsegment. That is, the edge $s$ is replaced by four edges obtained by adding three new vertices; see Figure~\ref{fig:fractal}.
 
The curve that is obtained by applying this construction indefinitely is the fractal known as the Koch curve~\cite{Koch1904}; it is one of the three curves forming the Koch snowflake. It is well known that the length of the Koch curve is unbounded, that is, for every $l > 0$, there exists an integer $n$, such that the length of $F_n$, i.e., the sum of its edge lengths, is greater than $l$.

\begin{figure}[ht]
    \centering
    \includegraphics[width=0.5\linewidth]{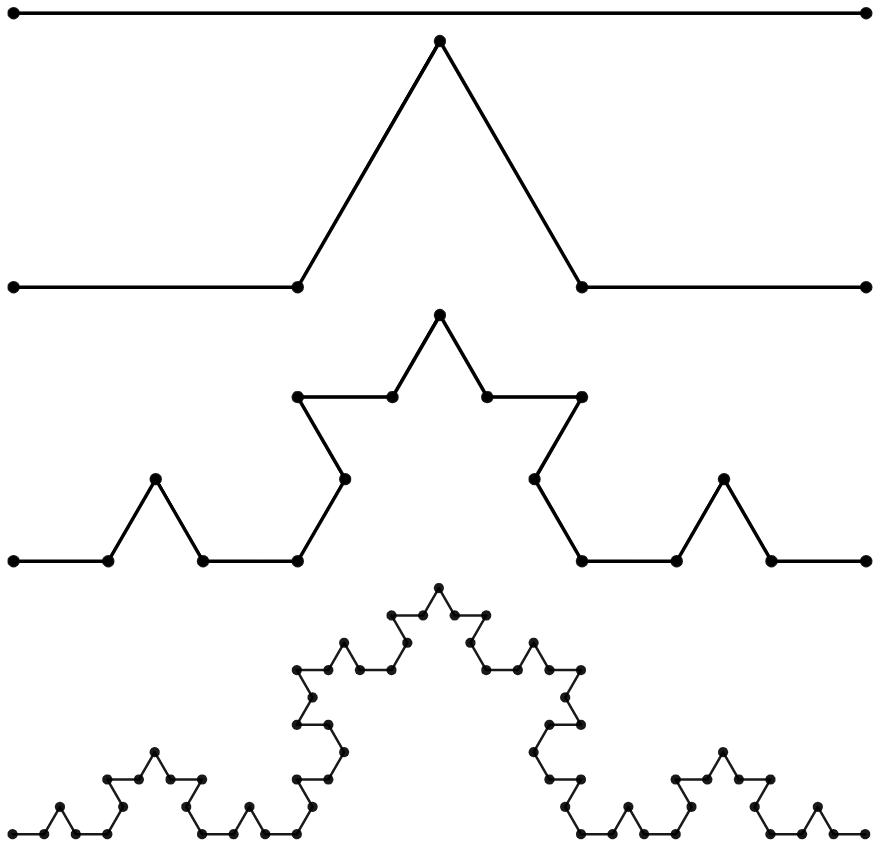}
    \caption{The graphs $F_0$, $F_1$, $F_2$ and $F_3$.}
    \label{fig:fractal}
\end{figure}

Since the length of the path between the extreme vertices of $F_n$, i.e., between its vertices at $(0,0)$ and $(1,0)$, can be made arbitrarily long, we conclude that for any $t \ge 1$, there exists an integer $n$, such that $F_n$ is not a $t$-spanner.

We next show that there exists a constant $c > 0$, such that for any $n \ge 0$, the graph $F_n$ is a $c$-Hausdorff spanner. We actually show that 6 is such a constant.

\paragraph*{Notation and definitions.}
Let $F_n = (V,E)$.
The \emph{level} of a vertex $v \in V$, denoted $l(v)$, is the smallest index $0 \le i \le n$, such that $v$ is already a vertex in $F_i$. If $l(v) \le i-1$, we write $l(v)=i^-$.

For a pair of vertices $u,v\in V$, we denote the path between $u$ and $v$ by $P(u,v)=(u,v_1,v_2,\ldots,v)$ and its corresponding sequence of levels by $P_l(u,v)$, that is, $P_l(u,v)=(l(u),l(v_1),l(v_2),\ldots,l(v))$. The \emph{level} of $u,v$, denoted $l(u,v)$, is now the smallest level $i$ such that there are at least two elements in the sequence $P_l(u,v)$ that are smaller or equal to $i$. For example, if $P_l(u,v)=(1,3,3,3,2,3,3,3)$, then $l(u,v)=2$.

The path from the leftmost vertex to the rightmost vertex induces a natural order on the vertices of $F_n$. We say that vertex $u$ of $F_n$ precedes/succeeds vertex $v$ of $F_n$ if $u$ appears before/after $v$ in this path.

Next, we define the \emph{bounding rectangle} of three consecutive vertices $v_1,v_2,v_3 \in V$; see Figure~\ref{fig:bounding_rect_def}. If the angle between $\overline{v_1v_2}$ and $\overline{v_2v_3}$ is $240^{\circ}$, then the bounding rectangle of $v_1,v_2,v_3$ is the rectangle such that (i) $v_1v_3$ is one of its diagonals and (ii) if $l(v_1) < l(v_3)$, the edge $\overline{v_1v_2}$ is contained in one of its long edges. Otherwise (i.e., $l(v_3) < l(v_1)$), the edge $\overline{v_2v_3}$ is contained in one of its long edges. If the angle between $\overline{v_1v_2}$ and $\overline{v_2v_3}$ is $60^{\circ}$, then the bounding rectangle of $v_1,v_2,v_3$ is the rectangle such that (i) one of its edges is $\overline{v_1v_3}$, and $v_2$ is on the opposite edge. Notice that the bounding rectangle is not always parallel to the axes, rather it is parallel to either $\overline{v_1v_2}$ or $\overline{v_2v_3}$, in the former case, or to the line segment $\overline{v_1v_3}$, in the latter case.

\begin{figure}[htb]
    \centering
    \includegraphics[width=0.5\linewidth]{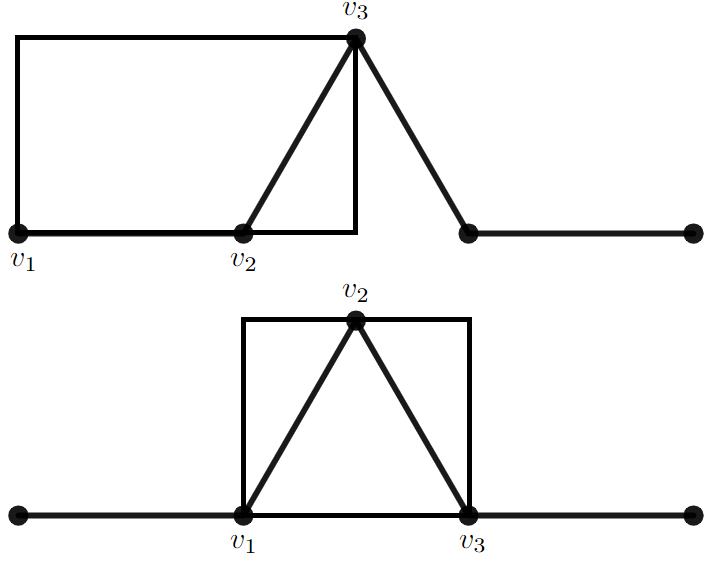}
    \caption{Top: The bounding rectangle when the angle between $\overline{v_1v_2}$ and $\overline{v_2v_3}$ is $240^{\circ}$.
    Bottom: The bounding rectangle when the angle between $\overline{v_1v_2}$ and $\overline{v_2v_3}$ is $60^{\circ}$}.
    \label{fig:bounding_rect_def}
\end{figure}

\begin{observation}
In the graph $F_n$, for any two adjacent vertices, at least one of them is of level $n$.
\end{observation}

\begin{observation}
    \label{obs:3_i_between_i-}
In the graph $F_n$ and for $1 \le i \le n$, let $v_1,v_2$ be two vertices that are adjacent to each other in the graph $F_{i-1}$. Then there are exactly three vertices of level $i$ between $v_1$ and $v_2$.

\end{observation}

\begin{observation}
    \label{obs:path_in_rectangle}
In the graph $F_n$ and for $1 \le i \le n$, let $v_1,v_2,v_3$ be three vertices that are consecutive in the graph $F_i$ and let $R$ be their bounding rectangle (with respect to $F_i$). Then $P(v_1,v_3) \subseteq R$, that is, the path between $v_1$ and $v_3$ in $F_n$ is contained in $R$.
\end{observation}

\begin{observation}
    \label{obs:max_distance_in_rectangle}
Consider the rectangle $R$ from Observation~\ref{obs:path_in_rectangle}. Then, the length of its diagonal is at most
$$
\sqrt{\left(\frac{3}{2}\left(\frac{1}{3}\right)^i\right)^2+\left(\frac{\sqrt{3}}{2}\left(\frac{1}{3}\right)^i\right)^2}=\frac{\sqrt{3}}{3^i} \, ,
$$
which is the length of its diagonal assuming it is the longer of the two possible rectangles, see Figure~\ref{fig:bounding_rect_def} (top).
\end{observation}

We now bound (from above) the Hausdorff distance (denoted $d_H$) between a path $P(u,v)$ in $F_n$ and the line segment $\overline{uv}$ as a function of $l(u,v)$.

\begin{lemma}
    \label{lem:Hausdorff_upper_bound}
    Consider the graph $F_n=(V,E)$ and let $u,v\in V$ such that $l(u,v)=i$,
    then $d_H(P(u,v),\overline{uv}) \le \frac{\sqrt{3}}{3^{i-1}}$.
\end{lemma}

\begin{proof}
%%%\matya{need to address some extreme cases.}
We prove the lemma for $i \ge 2$; for $i<2$ it is immediate.
We first identify three vertices $v_1,v_2,v_3$ that are consecutive in the graph $F_{i-1}$, such that their bounding rectangle (with respect to $F_{i-1}$) contains $P(u,v)$.

Since $l(u,v)=i$, there is at most one vertex in $P(u,v)$ of level $i^-$. If there is such a vertex, we set $v_2$ to be this vertex, set $v_1$ to be the the first vertex of level $i^-$ when moving from $u$ leftwards (i.e., towards the vertex at $(0,0)$), and set $v_3$ to be the first vertex of level $i^-$ when moving from $v$ rightwards (i.e., towards the vertex at $(1,0)$).

If all vertices in $P(u,v)$ are of level at least $i$, then we set $v_1$ to be the first vertex of level $i^-$ when moving from $u$ leftwards, and set $v_2$ and $v_3$ to be the first and second vertices, respectively, when moving from $v$ rightwards.

Clearly $v_1,v_2,v_3$ are consecutive vertices in $F_{i-1}$, and let $R$ be their bounding rectangle (with respect to $F_{i-1}$). By Observation~\ref{obs:path_in_rectangle}, the path $P(v_1,v_3)$ in $F_n$ is contained in $R$, and therefore so is the path $P(u,v)$ (since $P(u,v)$ is contained in $P(v_1,v_3)$). Clearly, the segment $\overline{uv}$ is also contained in $R$, since $R$ is convex.

Finally, since both $P(u,v)$ and $\overline{uv}$ are contained in $R$, the maximum distance between a point on $P(u,v)$ and a point on $\overline{uv}$ is at most the length of $R$'s diagonal, which according to Observation~\ref{obs:max_distance_in_rectangle} is at most $\frac{\sqrt{3}}{3^{i-1}}$. We thus conclude that $d_H(P(u,v),\overline{uv}) \le \frac{\sqrt{3}}{3^{i-1}}$.
\end{proof}

Next, we bound $d(u,v)$ for vertices $u$ and $v$ of $F_n$ (from below) as a function of $l(u,v)$.

\begin{lemma}
    \label{lem:distance_lower_bound}
    Consider the graph $F_n=(V,E)$ and let $u,v\in V$ such that $l(u,v)=i$,
    then $d(u,v) \geq \frac{1}{3^i} \cdot \frac{\sqrt{3}}{2}$.
\end{lemma}

\begin{proof}
Let $u'$ ($v'$) be the level-$i^-$ vertex of $F_n$ that precedes $u$ (succeeds $v$).
Since $l(u,v)=i$, there is at most one vertex in $P(u,v)$ of level $i^-$. We distinguish between three cases.

{\bf Case 1: There is no vertex of level $i^-$ in $P(u,v)$.}\\
By Observation~\ref{obs:3_i_between_i-}, there are exactly three vertices of level $i$ between $u'$ and $v'$. Moreover, at least two of them are in $P(u,v)$ (since $l(u,v) = i$ and there is no vertex of level $i^-$ in $P(u,v)$).
Assume, without loss of generality, that the middle and the right of these level-$i$ vertices are in $P(u,v)$.

We draw two parallel lines as depicted in Figure~\ref{fig:parallel_lines_case_1}. The first line passes through the left and middle level-$i$ vertices, and the second line passes through the right level-$i$ vertex and is parallel to the first line. Next, we observe that $u$ lies on one side of these lines and $v$ lies on the other side, and therefore the distance between $u$ and $v$ is at least the height of the level-$i$ triangle, formed by the three level-$i$ vertices. It is easy to verify that this height is $\frac{1}{3^i} \cdot \frac{\sqrt{3}}{2}$. We conclude that $d(u,v) \geq \frac{1}{3^i} \cdot \frac{\sqrt{3}}{2}$. 

{\bf Case 2: There is a $60^\circ$-vertex $w$ of level $i^-$ in $P(u,v)$.}\\
Since $l(u,v)=i$, there is at least one vertex of level $i$ in $P(u,v)$, and without loss of generality, we assume there is such a vertex that succeeds $w$.

We draw two parallel lines as depicted in Figure~\ref{fig:parallel_lines_case_2}. The first line passes through $w$ and the first level-$i$ vertex that precedes $w$. The second line is parallel to the first one and passes through the first level-$i$ vertex that succeeds $w$; by our assumption this vertex is in $P(u,v)$. As in the previous case, we observe that $u$ lies on one side of these lines and $v$ lies on the other side, and therefore the distance between $u$ and $v$ is at least the distance between the lines, which is the height of the level-$i$ triangle. We conclude that $d(u,v) \geq \frac{1}{3^i} \cdot \frac{\sqrt{3}}{2}$.

{\bf Case 3: There is a $240^\circ$-vertex $w$ of level $i^-$ in  $P(u,v)$.}
As in the previous case, there is at least one vertex of level $i$ in $P(u,v)$, and without loss of generality, we assume there is such a vertex that succeeds $w$.

We draw two parallel lines as depicted in Figure~\ref{fig:parallel_lines_case_3}. The first line passes through $w$ and the middle of the three level-$i$ vertices that lie between $v'$ and $w$. The second line is parallel to the first one and passes through the level-$i$ vertex that immediately succeeds $w$; by our assumption this vertex is in $P(u,v)$. Again, we observe that $u$ lies on one side of these lines and $v$ lies on the other side, and therefore the distance between $u$ and $v$ is at least the distance between the lines, which is the distance between two adjacent vertices in $F_i$. We conclude that $d(u,v) \ge \frac{1}{3^i} \ge \frac{1}{3^i} \cdot \frac{\sqrt{3}}{2}$.
\end{proof}

\begin{figure}[!htb]
    
\begin{subfigure}{\textwidth}
    \centering
    \includegraphics[width=0.45\textwidth]{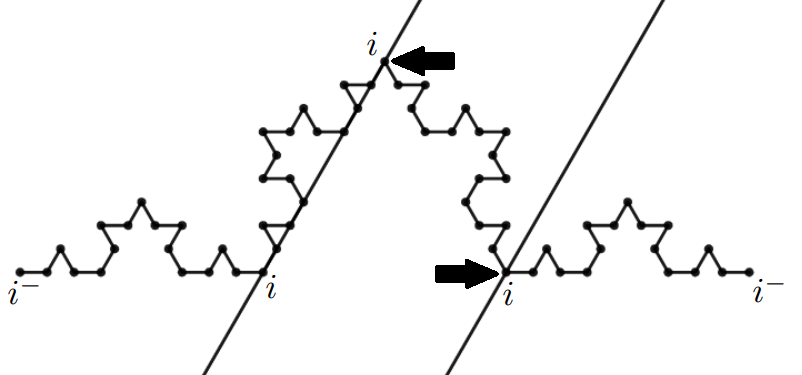}
    \caption{Case 1: There is no vertex of level $i^-$ in $P(u,v)$.}
    \label{fig:parallel_lines_case_1}
\end{subfigure}

\begin{subfigure}{\textwidth}
    \centering
    \includegraphics[width=0.55\textwidth]{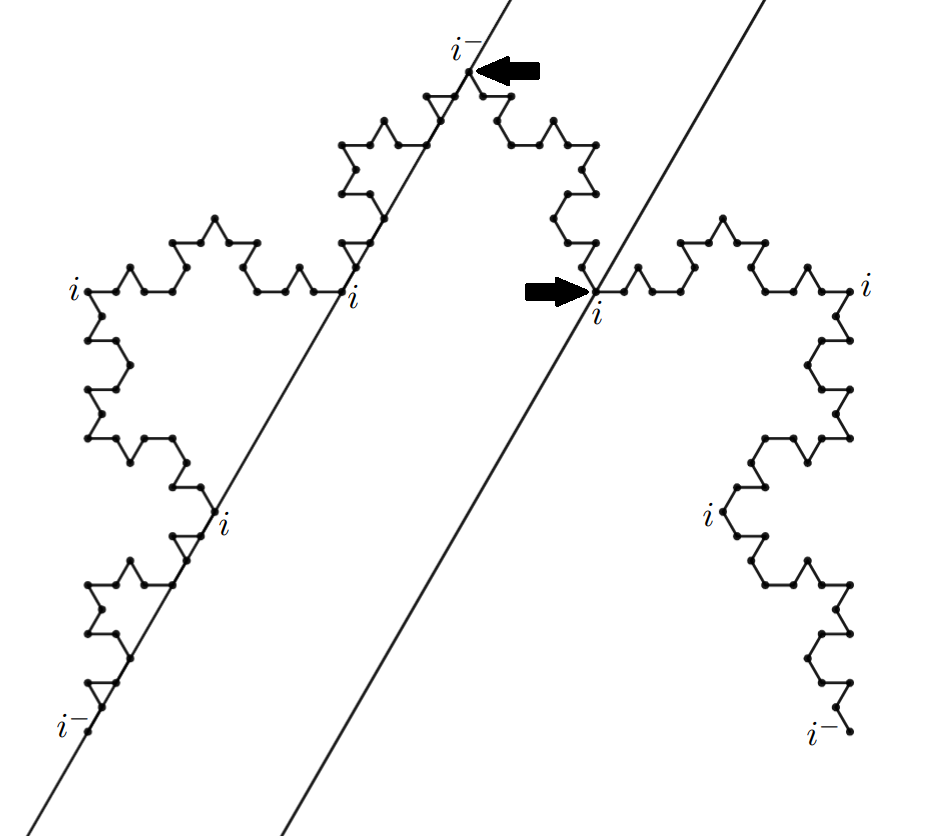}
    \caption{Case 2: There is a $60^\circ$-vertex of level $i^-$ in $P(u,v)$.}
    \label{fig:parallel_lines_case_2}
\end{subfigure}

\begin{subfigure}{\textwidth}
    \centering
    \includegraphics[width=0.85\textwidth]{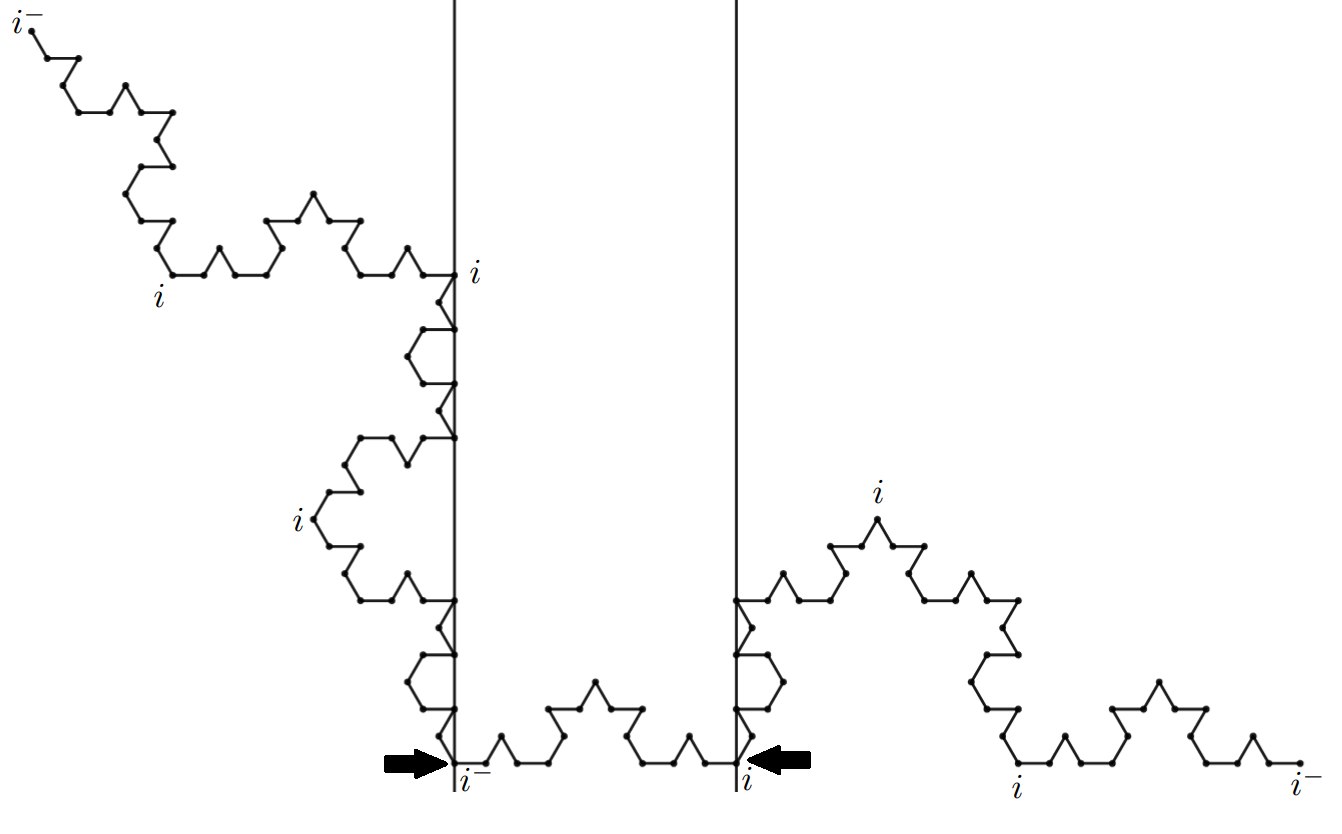}
    \caption{Case 3: There is a $240^\circ$-vertex of level $i^-$ in $P(u,v)$.}
    \label{fig:parallel_lines_case_3}
\end{subfigure}

\caption{Proof of Lemma~\ref{lem:distance_lower_bound}. The arrows mark the vertices that are known to be in the path $P(u,v)$.}
\label{parallel_lines}
\end{figure}

We are now ready to prove Theorem~\ref{thm:Hausdorff_not_Euclid}.
We first show that for any $n \ge 0$, the graph $F_n$ is a 6-Hausdorff spanner. Let $u,v$ be two vertices of $F_n$ and set $i=l(u,v)$. Then,
on the one hand by Lemma~\ref{lem:Hausdorff_upper_bound}, $d_H(P(u,v),\overline{uv}) \le \frac{\sqrt{3}}{3^{i-1}}$, and on the other hand by Lemma~\ref{lem:distance_lower_bound}, $d(u,v) \ge \frac{1}{3^i} \cdot \frac{\sqrt{3}}{2}$. Therefore, 
$$
\frac{d_H(P(u,v),\overline{uv})} {d(u,v)} \le \frac{\frac{\sqrt{3}}{3^{i-1}}}{\frac{1}{3^i}\cdot \frac{\sqrt{3}}{2}}=6 \, .
$$

Next, let $t \ge 1$. Then, there exists an integer $n > 1$, such that the length of the path $P$ between the extreme vertices of $F_n$ (i.e., the vertices at $(0,0)$ and $(1,0)$) is greater than $t$, and therefore $F_n$ is not a $t$-spanner (since the length of $P$ over the length of the segment between the extreme vertices of $F_n$ is  simply the length of $P$).

%%%%%%%%%%%%%%%%%%%%%%%%%%%%
%%%%% Fr{\'e}chet spanners
%%%%%%%%%%%%%%%%%%%%%%%%%%%%
\section{Fr{\'e}chet spanners}
\label{sec:Frechet}
\subsection{$t$-spanners are $\eps$-Fr{\'e}chet-spanners}

In this section we show that a $t$-spanner is an $\eps$-Fr{\'e}chet-spanner, for $\eps = f(t)$.

Let $H$ be a $t$-spanner, $t > 1$, and let $u,v \in S$. We assume, with loss of generality, that $d(u,v)=1$ and that $u=(0,0)$ and $v=(1,0)$. Let $P(u,v)$ be a $t$-path in $H$ between $u$ and $v$. Then, $P(u,v)$'s length is at most $t$.

We first prove that $P(u,v)$ is an $\eps$-Fr{\'e}chet-path, for $\eps=\frac{t}{2}$. This bound is useful when $t$ is `large', but, since it is never smaller than $1/2$, it is less useful when $t$ approaches 1, in which case we would like to show (if possible) that $\eps$ approaches 0. To address this issue, we prove a second bound on $\eps$, which is better when $1 < t < 2$. Specially, we prove that $P(u,v)$ is also an $\eps$-Fr{\'e}chet-path, for $\eps=\frac{\sqrt{t^2-t}}{\sqrt{2}}$.

\subsubsection{Bound 1 --- The bound for large $t$}

\begin{lemma}
\label{lem:bound_for_large_t}
    $P(u,v)$ is a $\frac{t}{2}$-Fr{\'e}chet-path.
\end{lemma}
\begin{proof}
Let $p$ be the middle point of $P(u,v)$, that is, the distance from $u$ to $p$ (through $P(u,v)$) is equal to the distance from $p$ to $v$ (through $P(u,v)$).
Consider a dog and its owner, both walking from $u$ to $v$, where the dog is walking along $P(u,v)$ and the owner is walking along $\overline{uv}$. Their hike consists of three stages. In the first, the owner is at $u$ and the dog advances to $p$; in the second, the owner advances from $u$ to $v$, while the dog stays at $p$; and in the third, the owner is at $v$ and the dog advances from $p$ to $v$.

We show that at any point along their hike, the distance between the dog and its owner does not exceed $t/2$. Indeed, let $p^-$ be a point on $P(u,v)$, anywhere between $u$ and $p$. Then, the distance from $u$ to $p^-$ (through $P(u,v)$) is at most $t/2$, and therefore $d(p^-,u) \le t/2$. Similarly, let $p^+$ be a point on $P(u,v)$, anywhere between $p$ and $v$. Then, the distance from $p^+$ to $v$ (through $P(u,v)$) is at most $t/2$, and therefore $d(p^+,v) \le t/2$. Finally, let $q$ be a point on $\overline{uv}$. We need to show that $d(p,q) \le t/2$. Indeed, if $q_x \le p_x$, then %%%$d(p,q) = \sqrt{(p_x - q_x)^2 + p_y^2} \le \sqrt{p_x^2 + p_y^2} = d(p,u) \le t/2$,
$d(p,q) \le d(p,u) \le t/2$,
and if $q_x \ge p_x$, then $d(p,q) \le d(p,v) \le t/2$.
\end{proof}

\begin{corollary}
\label{cor:bound1}
Any $t$-spanner is an $\eps$-Fr{\'e}chet-spanner, for $\eps=\frac{t}{2}$.
\end{corollary}

\subsubsection{Bound 2 --- The bound for small $t$\\}

\hspace{-2.75mm} Recalling Observation~\ref{obs:t-path_in_ellipse}, we observe that
\begin{observation}\label{obs:t-path_xy}
    Let $p$ be any point on $P(u,v)$, then
(i) $-\frac{t-1}{2} \le p_x \le \frac{t-1}{2}$ and 
(ii) $-\frac{\sqrt{t^2-1}}{2} \le p_y \le \frac{\sqrt{t^2-1}}{2}$.
\end{observation}

\begin{lemma}
\label{lem:bound_for_small_t}
    P(u,v) is a $\frac{\sqrt{t^2-t}}{\sqrt{2}}$-Fr{\'e}chet-path.
\end{lemma}
\begin{proof}
Consider a dog and its owner, both walking from $u$ to $v$, where the dog is walking along $P(u,v)$ and the owner is walking along $\overline{uv}$. We  denote the location of the dog at time $0 \le \tau \le 1$ by $P(\tau)$, where $P(0)=u$, $P(1)=v$, and, for any $\tau_1 < \tau_2$, the point $P(\tau_2)$ does not precede the point $P(\tau_1)$ (on $P(u,v)$). Moreover, we denote the $x$-coordinate of the rightmost point visited by the dog by time $\tau$ by $x_\tau$, that is, $x_{\tau} = \max \{P(\tau')_x \mid 0 \le \tau' \le \tau \}$.

The location of the person is determined by the location of the dog. More precisely, at time $\tau$ the person is at $(\max \{x_{\tau} - \frac{t-1}{2}, 0\},0)$. Finally, if at time $\tau=1$ the person is not yet at $v$, then she advances directly to $v$. Clearly, the person never moves backwards, since the function $x_\tau$ is non-decreasing.

We now prove that the distance between the dog and its owner never exceeds $\frac{\sqrt{t^2-t}}{\sqrt{2}}$.
Let $0 \le \tau \le 1$. We distinguish between three cases.

{\bf Case 1: $\max \left\{x_\tau-\frac{t-1}{2},0\right\} \le  P(\tau)_x$.}\\
\begin{equation*}
\begin{split}
d(P(\tau), & (\max \{x_\tau - \frac{t-1} {2},0\}, 0)) \le 
\sqrt{(P(\tau)_x-(x_{\tau}-\frac{t-1}{2}))^2+P(\tau)_y^2} \\
& \le
\sqrt{(x_{\tau}-(x_\tau-\frac{t-1}{2}))^2+P(\tau)_y^2} \le
\sqrt{\left(\frac{t-1}{2}\right)^2+\left(\frac{\sqrt{t^2-1}}{2}\right)^2} \\
& =
\sqrt{\frac{t^2-2t+1}{4}+\frac{t^2-1}{4}} =
\sqrt{\frac{2t^2-2t}{4}} =
\frac{\sqrt{t^2-t}}{\sqrt{2}} \, ,
\end{split}
\end{equation*}
where the first inequality is true, since, if $\max \{x_\tau-\frac{t-1}{2},0\} \ne x_\tau-\frac{t-1}{2}$, then $x_\tau-\frac{t-1}{2} < 0$ and both $P(\tau)_x$ and $(P(\tau)_x - (x_{\tau} - \frac{t-1}{2}))$ are non-negative where the latter is greater than the former. 

{\bf Case 2: $x_{\tau}-\frac{t-1}{2} > 0$ and $P(\tau)_x < x_{\tau}-\frac{t-1}{2}$.}\\
We first observe that $P(\tau)_x \ge x_{\tau}-(t-1)$, since $t \ge x_{\tau}+(x_{\tau}-P(\tau)_x)+(1-P(\tau)_x) \ge x_{\tau}+(1-P(\tau)_x)$. Now,
\begin{equation*}
\begin{split}
d(P(\tau), & (\max\{x_{\tau}-\frac{t-1}{2},0\},0))=
d(P(\tau),(x_{\tau}-\frac{t-1}{2},0)) \\
& =
\sqrt{(P(\tau)_x-(x_{\tau}-\frac{t-1}{2}))^2+P(\tau)_y^2} =
\sqrt{((x_{\tau}-\frac{t-1}{2})-P(\tau)_x)^2+P(\tau)_y^2} \\
& \le
\sqrt{((x_{\tau}-\frac{t-1}{2})-(x_{\tau}-(t-1)))^2+P(\tau)_y^2} \\
& \le
\sqrt{\left(\frac{t-1}{2}\right)^2+\left(\frac{\sqrt{t^2-1}}{2}\right)^2}=\frac{\sqrt{t^2-t}}{\sqrt{2}} \, .
\end{split}
\end{equation*}

   {\bf Case 3: $x_{\tau}-\frac{t-1}{2} \le 0$ and $P(\tau)_x<0$.}\\
   By Observations~\ref{obs:t-path_xy}, we have $-\frac{t-1}{2}\leq P(\tau)_x<0$.\\
\begin{equation*}
\begin{split}
    d(P(\tau), & (\max\{x_{\tau}-\frac{t-1}{2},0\},0))=
    d(P(\tau),(0,0))=
    \sqrt{P(\tau)_x^2+P(\tau)_y^2} \\
    & \le
    \sqrt{\left(-\frac{t-1}{2}\right)^2+\left(\frac{\sqrt{t^2-1}}{2}\right)^2} =
    \frac{\sqrt{t^2-t}}{\sqrt{2}} \, .
\end{split}
\end{equation*}

We have shown that in all cases the distance between the dog and its owner is at most $\frac{\sqrt{t^2-t}}{\sqrt{2}}$. Moreover, if the dog reaches $v$ first, then during the last part of the person's hike, this distance only decreases.
We thus conclude that P(u,v) is a $\frac{\sqrt{t^2-t}}{\sqrt{2}}$-Fr{\'e}chet-path.

\end{proof}

\begin{corollary}
\label{cor:bound2}
Any $t$-spanner is an $\eps$-Fr{\'e}chet-spanner, for $\eps=\frac{\sqrt{t^2-t}}{\sqrt{2}}$.
\end{corollary}

\subsection{$\eps$-Fr{\'e}chet-spanners are not necessarily $t$-spanners}
Consider the graph $F_n$, defined in Section~\ref{sec:Hausdorff_not_Euclid}. One can prove that $F_n$ is a $c$-Fr{\'e}chet-spanner, for $c \le 6$, in essentially the same way as we proved that it is a $c$-Hausdorff-spanner, for $c \le 6$. More precisely, referring to Lemma~\ref{lem:Hausdorff_upper_bound}, since both $P(u,v)$ and $\overline{uv}$ are contained in the rectangle $R$, the Fr{\'e}chet distance between them is at most the length of $R$'s diagonal, that is, $d_F(P(u,v),\overline{uv}) \le \frac{\sqrt{3}}{3^{i-1}}$. We thus conclude that

\begin{theorem}
There exists a constant $c > 0$, such that for any $t>1$, one can construct a graph that is a $c$-Fr{\'e}chet-spanner and is not a $t$-spanner.
\label{thm:Frechet_not_Euclid}
\end{theorem}

\section{The bound for WSPD-based $t$-spanners}

We proved (Corollaries~\ref{cor:bound1}~and~\ref{cor:bound2}) that \emph{any} $t$-spanner of a planar point-set $S$ is an $\eps$-Fr{\'e}chet-spanner for $\eps = \min\{\frac{t}{2},\frac{\sqrt{t^2-t}}{\sqrt{2}}\}$. In other words, if we wish to construct an $\eps$-Fr{\'e}chet-spanner of $S$ for some $\eps > 0$, it is sufficient to construct a $t$-spanner of $S$ for an appropriate value of $t$. The resulting spanner will have some desirable properties, depending on $t$ and on the specific algorithm chosen to construct it. In general, it is beneficial to set $t$ to the largest value for which our bounds guarantee that we get an $\eps$-Fr{\'e}chet-spanner. However, with our current general bounds, this value is very small, when $\eps$ is close to zero. For example, to obtain an $\eps$-Fr{\'e}chet-spanner for $\eps=0.01$, we need to set $t$ to $\frac{1+\sqrt{1+8\eps^2}}{2}$, which is approximately $1.0002$.

It is likely that by considering a specific spanner-construction algorithm, one can get better bounds. Indeed, we derive below a better bound for the WSPD-based spanner, where the WSPD is computed using a split tree. Specifically, we show that any WSPD-based $t$-spanner is an $\eps$-Fr{\'e}chet-spanner for $\eps = \frac{t-1}{2t+2}$. Thus, in the example above, we may set $t$ to approximately $1.04$.  

In more detail, given a separation factor $s > 0$, we show that the WSPD-based spanner (using $s$ as the separation factor) is a $\frac{2}{s}$-Fr{\'e}chet-spanner. But, to ensure that the spanner is a $t$-spanner, we need to set $s=\frac{4t+4}{t-1}$~\cite{NS09}, yielding the claimed bound, i.e., $\frac{t-1}{2t+2}$.       

We now sketch the proof of the claim that the WSPD-based spanner (using $s$ as the separation factor) is a $\frac{2}{s}$-Fr{\'e}chet-spanner. We begin by recalling a few well-known definitions.

Given a real number $s>0$, a \emph{well-separated pair} with respect to $s$ is a pair of two finite point-sets $A,B \subset \reals^d$, such that there exist two disjoint $d$-dimensional balls $C_A$ and $C_B$ satisfying the following conditions: (i) $C_A$ and $C_B$ have the same radius $r$, (ii) $C_A$ contains the bounding box $R(A)$ of $A$, and $C_B$ contains the bounding box $R(B)$ of $B$, and (iii) The distance between $C_A$ and $C_B$ is at least $sr$.

Given a set $S \subset \reals^d$ of $n$ points and $s>0$, a \emph{well-separated pair decomposition} (WSPD) of $S$ with respect to $s$~\cite{CallahanK95}, is a sequence $\{A_1,B_1\}, \ldots,\{A_m,B_m\}$ of pairs of non-empty subsets of $S$, for some integer $m$, such that
(i) $A_i$ and $B_i$ are well-separated with respect to $s$, for $i=1,\ldots,m$, and
(ii) for any pair of points $p,q \in S$, there is exactly one pair $\{A_i,B_i\}$, such that either $p\in A_i$ and $q\in B_i$, or $q\in A_i$ and $p\in B_i$. See Section~\ref{sec:background} below for the definition of a split tree of $S$ and for the algorithm for computing a WSPD of $S$ using it.

Given a WSPD of $S$ with separation $s$, a \emph{WSPD-based} spanner has $S$ as its vertex set and $E=\{(a_1,b_1),\ldots,(a_m,b_m)\}$ as its edge set, where $a_i$ and $b_i$ are arbitrary representatives of $A_i$ and $B_i$, respectively. 

We need to show that a WSPD-spanner, $G$, is a $\frac{2}{s}$-Fr{\'e}chet-spanner, assuming the WSPD is computed using a \emph{split tree}. That is, given a pair of points $u,v \in S$, $u \ne v$, we need to show that there exists a path $P(u,v)$ in $G$ between $u$ and $v$, such $d_F(P(u,v),\overline{uv}) \le \frac{2}{s} \cdot d(u,v)$. Indeed, assume $u \in A_i$ and $v \in B_i$. We first observe that $d_F(\overline{a_ib_i},\overline{uv}) \le \frac{2}{s} \cdot d(u,v)$. This is obvious, since $d(a_i,u) \le 2r$ and $d(b_i,v) \le 2r$, and so $d_F(\overline{a_ib_i},\overline{uv}) \le 2r = \frac{2}{s} \cdot sr \le \frac{2}{s} \cdot d(u,v)$. Now, since the WSPD is computed using a split tree, there exists a path $P(u,a_i)$ in $G$ that passes only through points of $A_i$ and is therefore contained in $C_{A_i}$, and there exists a path $P(b_i,v)$ in $G$ that passes only through points of $B_i$ and is therefore contained in $C_{B_i}$. (For completeness, we include a proof of the latter assertion, see Lemma~\ref{lem:pq_path_in_C} below, but we believe that it may be already known.) We consider the path $P(u,v)$ that is obtained by concatenating the paths $P(u,a_i)$ and $P(b_i,v)$, and assume the dog is walking along $P(u,v)$ and its owner along $\overline{uv}$. The owner waits for the dog at $u$ until it reaches $a_i$, then they walk simultaneously along their respective segments, i.e., $\overline{a_ib_i}$ and $\overline{uv}$, and finally, while the owner is at $v$, the dog walks from $b_i$ to $v$. Clearly, $d_F(P(u,v),\overline{uv}) \le 2r = \frac{2}{s}\cdot sr \le \frac{2}{s}\cdot d(u,v)$.   

\begin{corollary}
Any WSPD-based $t$-spanner is an $\eps$-Fr{\'e}chet-spanner, for $\eps=\frac{t-1}{2t+2}$.
\end{corollary}

\subsection{Missing background}
\label{sec:background}

For a finite set of point $S \subseteq \reals^d$, we denote its bounding box by $R(S)$. Given such a set $S$, we define the \emph{split tree} $T$ of $S$, which is a rooted binary tree, as follows~\cite{NS09}. 
If $S$ consists of a single point, then $T$ consists of single node that stores that point. Assume that $|S| \ge 2$. Split $R(S)$ into two equal-size boxes, by cutting with a hyperplane orthogonal to $R(S)$'s longest side. Let $S_1$ and $S_2$ be the subsets of $S$ that are contained in these two new boxes. The split tree of $S$ consists of a root with two subtrees, which are the recursively defined split trees of $S_1$ and $S_2$, respectively. Each internal node $u$ of $T$ corresponds to a subset $S_u$ of $S$ and to a bounding box, i.e., $R(S_u)$.

A WSPD of $S$ (with separation $s$) is computed by calling ComputeWSPD($T$,$s$), which calls FindPairs($l(u)$,$r(u)$), for each internal node $u$ of $T$, where $l(u)$ and $r(u)$ are the left and right children of $u$, respectively. 
FindPairs($v$,$w$) proceeds as follows. If $S_v$ and $S_w$ are well-separated, then it returns the pair $\{S_v,S_w\}$. Otherwise, without loss of generality, assume that the longest side of $R(w)$ is longer than that of $R(v)$. FindPairs($v$,$w$) issues two recursive calls, namely, FindPairs($v$,$l(w)$) and FindPairs($v$,$r(w)$).

Notice that the subsets of the WSPD of $S$ are subsets of $S$ corresponding to nodes of $T$.

\begin{lemma}
\label{lem:pq_path_in_C}
The spanner $G$ has the following property. For any subset $C$ participating in the WSPD of $S$ (i.e., $C=A_i$ or $C=B_i$, for some $1 \le i \le m$), and for any two points $p,q \in C$, there exists a path $P$ in $G$ between $p$ and $q$ that passes only through points in $C$. 
\end{lemma}
\begin{proof}
The proof is by induction on the size of $C$. If $|C|=1$, then the lemma is trivially true. Assume it is true for any subset $C$ of size at most $n$, and consider a subset $C$ of size $n+1$ and two points $p,q \in C$. Let $u$ be the node of $T$, such that $C=S_u$.

If there exists a node $u'$ of $T$, such that $p,q \in S_{u'}$ and $|S_{u'}| \le n$, then from $T$'s definition, $u'$ is a descendant of $u$ and $S_{u'} \subset S_u$.
From the Induction hypothesis, we get that there exists a path $P$ in $G$ from $p$ to $q$ that passes only through points in $S_{u'}$, and therefore only through points in $S_u$.

Assume now that there is no node $u'$ as above. Then, either $p \in S_{l(u)}$ and $q \in S_{r(u)}$, or $q \in S_{l(u)}$ and $p \in S_{r(u)}$, and assume, without loss of generality, that $p \in S_{l(u)}$ and $q \in S_{r(u)}$. The WSPD construction algorithm, ComputeWSPD($T$,$s$), issues the call FindPairs($l(u)$,$r(u)$), which returns a collection of well-separated pairs, including a unique pair $\{A,B\}$ ($A \subset S_{l(u)}$ and $B \subset S_{r(u)}$), such that $p \in A$ and $q \in B$. Let $a$ and $b$ be the representatives of $A$ and $B$, respectively. Then, by the induction hypothesis, there is a path $P(p,a)$ in $G$ between $p$ and $a$ that passes only through points of $A$, and there is a path $P(q,b)$ in $G$ between $q$ and $b$ that passes only through points of $B$. Moreover, there is an edge between $a$ and $b$ in $G$. Therefore, by concatenating the paths $P(p,a)$ and $P(q,b)$, we get a path from $p$ to $q$ that passes only through points of $S_u$.          
\end{proof}

\medskip

\noindent
{\bf Final remark.}
In Sections~\ref{sec:Hausdorff}~and~\ref{sec:Frechet}, we have restricted our attention to the plane. That is, we showed that a $t$-spanner of a planar point-set $S$ is also an $f_1(t)$-Hausdorff spanner of $S$ and an $f_2(t)$-Fr{\'e}chet-spanner of $S$. However, this is true in any dimension $d$.
%Moreover, we did not try to optimize the functions (i.e. bounds) $f_1$ and $f_2$.

\medskip

%%%\paragraph*{Acknowledgment.}
\noindent
{\bf Acknowledgment.}
We thank Boris Aronov for helpful discussions on the connection between $t$-spanners and $\eps$-Hausdorff-spanners.


\begin{thebibliography}{1}


\bibitem{CallahanK95}
Paul B. Callahan and S. Rao Kosaraju.
\newblock A decomposition of multidimensional point sets with Applications to $k$-nearest-neighbors and $n$-body potential fields
\newblock {J. {ACM}}, 42(1):67--90, 1995.

\bibitem{Frechet1906}
Maurice Fr{\'e}chet.
\newblock Sur quelques points du calcul fonctionnel.
\newblock {Rendiconti del Circolo Matematico di Palermo}, 22:1--72, 1906.

\bibitem{Hausdorff1914}
Felix Hausdorff.
\newblock Grundz{\"u}ge der Mengenlehre.
\newblock Veit \& Comp., 1914.

\bibitem{Koch1904}
Helge von Koch. 
\newblock Sur une courbe continue sans tangente, obtenue par une construction g{\'e}om{\'e}trique {\'e}l{\'e}mentaire.
\newblock {\em Arkiv f{\"o}r matematik, astronomi och fysik}, 1:681--704, 1904.

\bibitem{NS09}
Giri Narasimhan and Michiel Smid.
\newblock {\em Geometric Spanner Networks}.
\newblock Cambridge University Press, 2009.

\end{thebibliography}
\end{document}